\PassOptionsToPackage{unicode}{hyperref}
\PassOptionsToPackage{hyphens}{url}
\PassOptionsToPackage{dvipsnames,svgnames,x11names}{xcolor}
\documentclass[
  12pt]{article}
\usepackage{comment}
\usepackage{amsmath,amssymb,amsthm}
\usepackage{iftex}
\ifPDFTeX
  \usepackage[T1]{fontenc}
  \usepackage[utf8]{inputenc}
  \usepackage{textcomp} 
\else 
  \usepackage{unicode-math}
  \defaultfontfeatures{Scale=MatchLowercase}
  \defaultfontfeatures[\rmfamily]{Ligatures=TeX,Scale=1}
\fi
\usepackage{lmodern}
\ifPDFTeX\else  
\fi
\IfFileExists{upquote.sty}{\usepackage{upquote}}{}
\IfFileExists{microtype.sty}{
  \usepackage[]{microtype}
  \UseMicrotypeSet[protrusion]{basicmath} 
}{}
\makeatletter
\@ifundefined{KOMAClassName}{
  \IfFileExists{parskip.sty}{%
    \usepackage{parskip}
  }{
    \setlength{\parindent}{0pt}
    \setlength{\parskip}{6pt plus 2pt minus 1pt}}
}{
  \KOMAoptions{parskip=half}}
\makeatother
\usepackage{xcolor}
\makeatletter
\ifx\paragraph\undefined\else
  \let\oldparagraph\paragraph
  \renewcommand{\paragraph}{
    \@ifstar
      \xxxParagraphStar
      \xxxParagraphNoStar
  }
  \newcommand{\xxxParagraphStar}[1]{\oldparagraph*{#1}\mbox{}}
  \newcommand{\xxxParagraphNoStar}[1]{\oldparagraph{#1}\mbox{}}
\fi
\ifx\subparagraph\undefined\else
  \let\oldsubparagraph\subparagraph
  \renewcommand{\subparagraph}{
    \@ifstar
      \xxxSubParagraphStar
      \xxxSubParagraphNoStar
  }
  \newcommand{\xxxSubParagraphStar}[1]{\oldsubparagraph*{#1}\mbox{}}
  \newcommand{\xxxSubParagraphNoStar}[1]{\oldsubparagraph{#1}\mbox{}}
\fi
\makeatother

\usepackage{longtable,booktabs,array}
\usepackage{calc} 
\usepackage{etoolbox}
\makeatletter
\patchcmd\longtable{\par}{\if@noskipsec\mbox{}\fi\par}{}{}
\makeatother
\IfFileExists{footnotehyper.sty}{\usepackage{footnotehyper}}{\usepackage{footnote}}
\makesavenoteenv{longtable}
\usepackage{graphicx}
\makeatletter
\def\maxwidth{\ifdim\Gin@nat@width>\linewidth\linewidth\else\Gin@nat@width\fi}
\def\maxheight{\ifdim\Gin@nat@height>\textheight\textheight\else\Gin@nat@height\fi}
\makeatother
\setkeys{Gin}{width=\maxwidth,height=\maxheight,keepaspectratio}
\makeatletter
\def\fps@figure{htbp}
\makeatother

\makeatletter
\@ifpackageloaded{caption}{}{\usepackage{caption}}
\AtBeginDocument{%
\ifdefined\contentsname
  \renewcommand*\contentsname{Table of contents}
\else
  \newcommand\contentsname{Table of contents}
\fi
\ifdefined\listfigurename
  \renewcommand*\listfigurename{List of Figures}
\else
  \newcommand\listfigurename{List of Figures}
\fi
\ifdefined\listtablename
  \renewcommand*\listtablename{List of Tables}
\else
  \newcommand\listtablename{List of Tables}
\fi
\ifdefined\figurename
  \renewcommand*\figurename{Figure}
\else
  \newcommand\figurename{Figure}
\fi
\ifdefined\tablename
  \renewcommand*\tablename{Table}
\else
  \newcommand\tablename{Table}
\fi
}
\@ifpackageloaded{float}{}{\usepackage{float}}
\floatstyle{ruled}
\@ifundefined{c@chapter}{\newfloat{codelisting}{h}{lop}}{\newfloat{codelisting}{h}{lop}[chapter]}
\floatname{codelisting}{Listing}

\makeatother
\makeatletter
\@ifpackageloaded{caption}{}{\usepackage{caption}}
\@ifpackageloaded{subcaption}{}{\usepackage{subcaption}}
\makeatother

\ifLuaTeX
  \usepackage{selnolig}  
\fi
\usepackage[]{natbib}
\usepackage{bookmark}

\IfFileExists{xurl.sty}{\usepackage{xurl}}{} 
\hypersetup{
  pdftitle={Title},
  pdfauthor={Author 1; Author 2},
  pdfkeywords={3 to 6 keywords, that do not appear in the title},
  colorlinks=true,
  linkcolor={blue},
  filecolor={Maroon},
  citecolor={Blue},
  urlcolor={Blue},
  pdfcreator={LaTeX via pandoc}}
\newcommand{\anon}{1}

\begin{document}
\newtheorem{definition}{Definition}[section]
\newtheorem{theorem}{Theorem}[section]
\newtheorem{prop}{Proposition}[section]
\newtheorem{corollary}{Corollary}[prop]
\newtheorem{lemma}{Lemma}[section]
\newtheorem{remark}{Remark}
\newcommand{\blind}{0}
\newcommand{\Expect}[1]{\mathbb{E}\left[ #1 \right]}
\newcommand{\av}{\vec{a}}
\newcommand{\diffd}{\textnormal{d}}
\newcommand{\deriv}[2]{\frac{\diffd #1}{\diffd #2}}
\newcommand{\pderiv}[2]{\frac{\partial #1}{\partial #2}}
\newcommand{\inprod}[2]{\left\langle #1,#2\right\rangle}
\newcommand{\abs}[1]{\left| #1 \right|}
\newcommand{\fourier}[1]{\mathcal{F}\left[#1\right]}
\newcommand{\fourieri}[1]{\mathcal{F}^{-1}\left[#1\right]}
\newcommand{\norm}[1]{\left|\left|  #1 \right| \right|}
\newcommand{\Ltwo}[2]{\mathcal{L}^2_{#2}\left( #1 \right)}
\newcommand{\LtwoStwo}{\mathcal{L}^2\left( \mathbb{S}^2 \right)}
\newcommand{\Disk}{\mathbb{D}}
\newcommand{\SPAN}[1]{\text{span}\left(#1\right)}
\newcommand{\Stwo}{\mathbb{S}^2}
\newcommand{\mat}[1]{\boldsymbol{#1}}
\newcommand{\Floor}[1]{\lfloor #1 \rfloor }
\newcommand{\Ceil}[1]{\lceil #1 \rceil }
\newcommand{\tr}{\text{tr}}
\newcommand\numberthis{\addtocounter{equation}{1}\tag{\theequation}}
\newcommand{\degrees}{^{\text{o}}}
\newcommand{\Var}[1]{\mathrm{Var}\left[ #1 \right]}
\newcommand{\Cov}[2]{\mathrm{Cov}\left(#1,#2\right)}
\newcommand{\R}{\mathbb{R}}
\newcommand{\intR}{\int_{\R}}
\newcommand{\intPR}{\int_0^{\infty}}
\newcommand{\intNR}{\int_{-\infty}^0}
\newcommand{\intLim}[2]{\int_{#1}^{#2}}
\newcommand{\Pois}[1]{\mathrm{Poisson}\left(#1\right)}
\newcommand{\Prob}[1]{\mathbb{P}\left(#1\right)}
\newcommand{\Norm}[2]{\mathcal{N}\left(#1,#2\right)}
\newcommand{\MVN}[2]{\mathcal{MVN}\left(#1,#2\right)}
\newcommand{\ai}{\{a_i\}_{i=1}^{n_1}}
\newcommand{\bi}{\{b_i\}_{i=1}^{n_2}}
\newcommand{\aij}{\{a_{ij}\}_{i=1}^{n_{1j}}}
\newcommand{\bij}{\{b_{ij}\}_{i=1}^{n_{2j}}}
\newcommand{\eye}{\mathrm{\mat{I}}}
\newcommand{\indep}{\perp\!\!\!\perp}
\newcommand{\rlang}{\texttt{R} }
\newcommand{\wvlet}{Solid Needlet }
\newcommand{\gbp}{generalized Beta-Prime }
\newcommand{\scr}[1]{\mathscr{#1}}
\newcommand{\sinc}[1]{\mathrm{sinc}\left(#1\right)}
\newcommand{\fraC}[2]{\frac{#1}{#2}}
\newcommand{\sneed}{\varrho}
\newcommand{\sign}[1]{\text{sgn}\left(#1 \right)}
\numberwithin{equation}{section}
\numberwithin{figure}{section}
\numberwithin{table}{section}

\newcommand{\diff}[2]{\frac{\partial #1}{\partial #2}}
\newcommand{\diffr}[1]{\diff{#1}{r}}
\newcommand{\diffth}[1]{\diff{#1}{\theta}}
\newcommand{\diffz}[1]{\diff{#1}{z}}

\newcommand{\vth}{V_{\theta}}

\newcommand{\twochoices}[2]{\left\{ \begin{array}{lcc}
        \displaystyle #1 \\ \vspace{-10pt} \\
        \displaystyle #2 \end{array} \right. } 

\newcommand{\threechoices}[3]{\left\{ \begin{array}{lcc}
        #1 \\ #2 \\ #3 \end{array} \right. }    

\newcommand{\fourchoices}[4]{\left\{ \begin{array}{lcc}
        #1 \\ #2 \\ #3 \\ #4 \end{array} \right. }      

\newcommand{\twovec}[2]{\left(\begin{array}{c} #1 \\ #2 \end{array}\right)}
\newcommand{\threevec}[3]{\left(\begin{array}{c} #1 \\ #2 \\ #3 \end{array}\right)}
\newcommand{\twomatrix}[4]{\left(\begin{array}{cc} #1 & #2 \\ #3 & #4 \end{array}\right)}

\def\spacingset#1{\renewcommand{\baselinestretch}%
{#1}\small\normalsize} \spacingset{1}


\if1\anon
{
  \title{\bf Calibration Assessment for Nested Uncertainty Sets}
  \author{Matthew LeDuc\hspace{.2cm}\\
    Department of Applied Mathematics, University of Colorado Boulder}
  \maketitle
} \fi

\if0\anon
{
  \bigskip
  \bigskip
  \bigskip
  \begin{center}
    {\LARGE\bf Probabilistic Assessment of Model Reliability Using a Multinomial-Dirichlet Model}
\end{center}
  \medskip
} \fi

\maketitle
\begin{abstract}
    We present a Bayesian framework for assessing the calibration of nested uncertainty sets from simulation studies. Calibration is typically evaluated by repeatedly simulating data from a known model and comparing empirical coverage probabilities with their nominal values, an approach that often requires hundreds or thousands of independent model evaluations to obtain reliable estimates. Such computational demands can be prohibitive when the forward models are expensive to evaluate, for example in data assimilation using a large Earth-system model, or may not be applicable, for exmaple when the output of the simulation is a single or multiple uncertainty sets rather than a full distribution. We present a method that leverages the nested structure of uncertainty sets across nominal coverage levels to assess the calibration of these models. The method, which is based on a Dirichlet-Multinomial model, allows for testing both exact and region-based notions of model calibration in an efficient manner using Bayes factors. Unlike methods that require a full predictive distribution, this method works for procedures that directly produce prediction or confidence intervals. We demonstrate the method on a data assimilation problem containing model form error and show that it is capable of detecting miscalibration even when there are only a limited number of simulation studies available.
\end{abstract}
\noindent%

\spacingset{1.8} 
\section{Introduction and Motivation}

The ability of a statistical model to accurately quantify uncertainty is as important as its ability to produce accurate point estimates. In applications such as robotics, weather forecasting, and remote sensing, decisions are often based not only on the predicted value of a quantity of interest, but also on the reliability of the accompanying uncertainty estimates. A model whose credible or prediction intervals are systematically too narrow or too wide can lead to overconfident or unnecessarily conservative decisions, even when its point predictions are accurate. This can be especially important in applications where safety is critical, such as drug testing or autonomous system design, where overconfident models can cause significant harm (\cite{valentin2026recipescalibrationcheckssafetycritical}). Consequently, assessing the calibration of probabilistic predictions is a fundamental component of model validation (\cite{brocker2007increasing,pmlr-v80-kuleshov18a,gebetsberger2017estimation}).

A standard approach for evaluating calibration is to perform a simulation study. Independent datasets are generated from a known data-generating mechanism, the model is fit to each dataset, and the proportion of credible intervals containing the true parameter is reported at one or more nominal coverage levels. If these empirical coverages are close to the nominal values, the model is deemed to be well calibrated. This approach is simple and widely used, but it has several important limitations. First, the reported empirical coverages are themselves estimated from a finite number of simulations and therefore possess non-negligible sampling uncertainty (\cite{deperettiGraphicalMethodsInvestigating2010}). Although this uncertainty can be quantified using classical binomial methods, it is rarely incorporated into a formal inferential framework for assessing calibration. Second, obtaining sufficiently precise estimates may require hundreds or thousands of independent model evaluations. Such simulation studies are computationally feasible for simple models, but become impractical for modern Bayesian and data assimilation algorithms whose forward models may require minutes or hours per evaluation. Finally, existing approaches such as the Probability Integral Transform (e.g. \cite{gneiting2007probabilistic}) are not applicable in settings where the full predictive distribution is not available, such as a model that generates confidence intervals for a parameter of interest. This suggests the need for calibration assessment tools that directly handle these data. While tools for adjusting these intervals to achieve calibration exist \citep{JMLR:v27:24-1139}, this method is intended to do inference about the calibration of a given model. 

To make these ideas precise, let $\pi(\theta\mid x)$ denote a distribution for a scalar parameter of interest. For a nominal coverage level $\alpha\in(0,1)$, let $I_\alpha(x)$ denote the corresponding $100\alpha\%$ credible interval. The true coverage probability is defined by
\begin{equation}
p(\alpha)=\Pr_{\theta_0,x}\left(\theta_0\in I_\alpha(x)\right),
\end{equation}
where the probability is taken over both the true parameter and the observed data. A model is said to be calibrated if $p(\alpha)=\alpha$
for every nominal level $\alpha$.

Suppose a simulation study consists of $n$ independent realizations $(\theta_0^{(j)},x^{(j)})$. The conventional estimate of the calibration function is obtained from the empirical coverage,
\begin{equation}
\hat p(\alpha)
=
\frac{1}{n}
\sum_{j=1}^{n}
\mathbf 1\left(
\theta_0^{(j)}
\in
I_\alpha(x^{(j)})
\right).
\end{equation}
While this estimator is unbiased, it does not by itself provide a principled framework for inference on calibration, particularly when the number of simulations is limited. Additionally, since the nestedness of the $I_{\alpha}(x^{(j)})$ induces correlation in the data, care should be taken when calibration is assessed simultaneously across multiple nominal levels.

In this paper we develop a Bayesian framework for inference on calibration curves. Rather than treating empirical coverages independently, we model the multinomial counts across nominal coverage levels using a Dirichlet--multinomial model, thereby accounting for the dependence induced by nested credible intervals. The resulting posterior distribution provides uncertainty quantification for the entire calibration curve and enables Bayes factor tests of both point and region hypotheses, including tolerance-band and one-sided notions of calibration. Building on the calibration checks of \cite{valentin2026recipescalibrationcheckssafetycritical}, the proposed framework replaces empirical coverage estimates with formal probabilistic inference while remaining computationally inexpensive. We illustrate the methodology on a data assimilation problem based on a Duffing oscillator and show that we are able to use this method to perform inference about the model calibration under model form error as well as error in the filter parameters, and that we are able to test whether the model is systematically under or over conservative, even with a relatively small number of simulations. Although motivated by computationally intensive data assimilation systems or methods that do not provide full predictive distributions, the proposed methodology is applicable to any simulation-based assessment of model calibration.

\section{Method}

Consider a predictive model which, upon receipt of data $x^{(j)}$ returns a predictive distribution for some parameter $\theta_0^{(j)}$. By any one of various means credible sets for the parameter $\theta_0$ can be derived from this distribution. 

Fix $d$ nominal levels
$0 < \alpha_1 < \alpha_2< ... < \alpha_d < 1$. Since the credible sets are nested,
$I_{\alpha_{i}}(x) \subset I_{\alpha_{i+1}}(x)$, each of $n$ independent simulation
replicates $(\theta_0^{(j)}, x^{(j)})$ falls into exactly one of $d+1$
mutually exclusive and exhaustive categories:
\begin{equation}
\begin{split}
    A_1 &: \theta_0 \in I_{\alpha_1}(x)\\
    A_2 &: \theta_0 \in I_{\alpha_2}(x) \setminus I_{\alpha_1}(x) \\
    \vdots \\
     A_{d+1} &: \theta_0 \notin I_{\alpha_d}(x) \\
\end{split}
\end{equation}
with "gap probabilities" $q_1, q_2,...,q_{d+1} \geq 0$ satisfying $\sum_{i=1}^{d+1}q_i = 1$.
The true coverage probabilities at the nominal levels are
\begin{equation}\label{eq:c_def}
    c_0=0, c_1 = q_1, c_2=q_1+q_2, ..., c_i = \sum_{i'=1}^iq_{i'}, c_{d+1}=1
\end{equation}
so that $c_i - c_{i-1} = q_i \geq 0$, guaranteeing $c_{i} \geq c_{i-1}$ automatically. The joint distribution of the $c_i$ given the data is derived in Appendix \ref{app:cidist}, but due to the formulation of the model it is more natural to work with $q_i$ instead. These probabilities sum to 1, and thus the vector $\vec{q}\in \Delta_d$, the simplex in $\mathbb{R}^{d+1}$. Any calibration criterion expressible in terms of affine inequalities on the cumulative coverages $\vec{c}$, such as the tolerance band criterion or testing for overconfidence, corresponds to a set of linear constraints on $\vec{q}$, so we lose little by this formulation.  

Let $n_i$ denote the observed counts in $A_i$, with
$\sum_i n_i = n$. Since the sets $A_i$ are disjoint and exhaustive, we can say that the quantities $n_i$ are jointly multinomially distributed as
\begin{equation}\label{eq:multinom_dist}
    n_1, n_2,...,n_{d+1} \mid q_1, q_2, ...,q_{d+1}
    \sim \mathrm{Multinomial}(n,\, q_1,q_2,...,q_{d+1})
\end{equation}
This serves as the starting point for our analysis of empirical reliability curves. 
We consider the Bayesian perspective on model calibration checking via Bayes factors. First, before seeing the results of the calibration test we place a Dirichlet prior on the probabilities $q_i$ so that
\begin{equation}
    \label{eq:dirichlet_prior}
     q_1, q_2,...,q_{d+1} \sim \mathrm{Dir}(\gamma_1, \gamma_2,...,\gamma_{d+1})
\end{equation}
which by conjugacy gives
\begin{equation}
    \label{eq:posterior_dist_q}
     q_1, q_2,...,q_{d+1}|n_1,n_2,...n_{d+1} \sim \mathrm{Dir}(a_1,a_2,...,a_{d+1})
\end{equation}
where $a_i=\gamma_i+n_i$.
Natural choices include the uniform prior $\gamma_i = 1$ for all $i$ or the Jeffreys prior $\gamma_i=1/2$. 

\begin{remark}
As $d \to \infty$ with the nominal levels becoming dense in $(0,1)$, the
Dirichlet model on the gap probabilities converges to a Dirichlet process
prior on the density $g$ in the representation $f(\alpha) = \int_0^\alpha
g(t)\, dt$. This provides a natural
nonparametric limit of the multinomial--Dirichlet approach and connects the
finite-level calibration assessment to a random function estimation problem for the
reliability curve, as done in isotonic regression (\cite{jiang2011smooth}).
\end{remark}

\subsection{Calibration Regions}

Now that the calibration data have been collected, we test the two hypotheses
\begin{equation}\label{eq:hypotheses}
    \begin{split}
        H_0: \vec{q}\in R\\
        H_1: \vec{q}\notin R
    \end{split}
\end{equation}
for some $R\subset \Delta_d$ using Bayes factors (\cite{kass1995bayesfactors, jarosz2014odds}). 
Because the hypothesis is specified by an arbitrary region, the framework immediately accommodates exact calibration, tolerance bands, one-sided calibration, and other application-specific notions of acceptable calibration. We will discuss the special case of exact calibration in the following section. In general, this form of the test, where we check whether the vector of gap probabilities lies in a set $R$, is most attractive for applications. Consider, for example, a safety-critical setting where the credible sets being too large is not a concern, but them being too small could have a significant drawback. One example of this is in autonomous systems, as too small of error estimates on the location of a robot could lead to it inadvertently being lead off a cliff, or healthcare where overconfidence in a novel treatment may cost people their lives \citep{valentin2026recipescalibrationcheckssafetycritical}. Additionally, exact calibration may not always be required, or different notions of calibration may be of interest \citep{resin2026hierarchies}. It may be acceptable to have the model nearly calibrated so that, for example, 
\begin{equation}
    \mathrm{max}\abs{\frac{c_i-\alpha_i}{\sqrt{\alpha_i(1-\alpha_i)}}} <\epsilon
\end{equation}
for some tolerance $\epsilon$. In that case, the test really only needs to check whether the model is acceptably close to calibrated rather than whether it is exactly calibrated. 

Letting $m_0(x)$ and $m_1(x)$ be the models under $H_0$ and $H_1$ respectively, the Bayes factor $BF_{10}$ is given by
\begin{equation}
    \label{eq:bayesfactor}
    BF_{10} = \frac{m_1(x)}{m_0(x)}
\end{equation}
where the models are evaluated at the observed data. This can be written in terms of integrals over the region of interest and its complement on the simplex as
\begin{equation}
    BF_{10}(R) = \frac{\int_{R^c} p(\vec{n}|\vec{q})\pi_{R^c}(\vec{r})\diffd \vec{q}}{\int_{R} p(\vec{n}|\vec{q})\pi_{R}(\vec{r})\diffd \vec{q}}
\end{equation}
where $\pi_{A}$ is the restriction of the prior to the set $A$ properly normalized. Letting 
\begin{equation}
        p_{\pi}(R)=\int_{R}\pi(\vec{q})\diffd \vec{q}
\end{equation}
be the prior probability that $\vec{q}\in R$ and 
\begin{equation}
    p_*(R)=\frac{\int_R p(\vec{n}|\vec{q})\pi(\vec{q})\diffd \vec{q}}{\int_{\Delta d} p(\vec{n}|\vec{q})\pi(\vec{q})\diffd \vec{q}}
\end{equation}
be the posterior probability, we can say that
\begin{equation}
    \label{eq:bf10_region}
    BF_{10}(R)=\frac{p_*(R)(1-p_{\pi}(R))}{p_{\pi}(R)(1-p_*(R))}
\end{equation}
which, since we can easily sample from the prior and posterior (assuming the prior is proper), can be calculated via Monte-Carlo.

While the calculation is straightforward, the interpretation is not.
In \cite{kass1995bayesfactors}, a table is provided for suggested interpretations of the Bayes factors (See Section 3.2 therein). The table suggests that $BF_{10}>3.2$ provides "substantial" evidence for the alternative hypothesis, while $BF_{10}>10$ and $BF_{10}>100$ provide "strong" and "decisive" evidence respectively. A Bayes factor less than 3.2 is considered "not worth more than a bare mention" and a Bayes factor less than 1 is considered evidence in favor of $H_0$, with strength of evidence arranged in the same manner. The textbook \cite{jeffreys1998theory} provides another similar set of interpretations. 

 The Bayes factor is known to be sensitive to the prior distribution and does not have an associated decision rule like a likelihood ratio test might. In addition, the Bayes factor is not defined if the prior is improper. For this reason, the prior distribution should be reported as well. This is especially important in the context we are considering, since when only a small number of model runs can be used to estimate calibration the results will be especially sensitive to the prior chosen. In a frequentist setting, we can do away with the prior, which motivates future work in that direction. The same geometric formulation naturally suggests likelihood-ratio tests under region constraints, although efficient finite-sample calibration of these tests is an interesting topic for future work.

\subsection{Exact Calibration}

A special case of the model described in the previous section occurs when $R=\{\vec{p}_0\}$ is a singleton set. In this case, we are able to derive a closed form of the Bayes factor. 
Under the model discussed above, the null model is 
\begin{equation}\label{eq:null_model}
    n_1,n_2,...,n_{d+1}|\vec{p}_0\sim \mathrm{Multinomial}(n,\, \vec{p}_0)
\end{equation}
while the alternative model $m_1$ is given by
\begin{equation}
    n_1,n_2,...,n_{d+1}|\vec{\gamma}\sim \mathrm{DirMult}\left(n, \vec{\gamma}\right)
\end{equation}
the Dirichlet-Multinomial distribution with probability mass function
\begin{equation}
    \label{eq:dirmult_pmf}
    p(\vec{x}|\vec{\gamma})=\frac{\Gamma\left(\sum_i \gamma_i\right)\Gamma(n+1)}{\Gamma\left(n+\sum_i\gamma_i\right)}\prod_{i=1}^{d+1}\frac{\Gamma(x_i+\gamma_i)}{\Gamma(\gamma_i)x_i!}
\end{equation}
Note that the parameter in the Dirichlet-Multinomial distribution for the calculation of the marginal likelihood of the model is the prior concentration vector $\vec{\gamma}$, rather than the posterior concentration vector $\vec{a}$.
So, the Bayes factor $BF_{10}(\vec{p}_0)$ for this test is given by
\begin{equation}
    \label{eq:bayesfactor_closedform}
    BF_{10}(\vec{p}_0) = \frac{\Gamma\left(\sum_i \gamma_i \right)}{\Gamma\left(n+\sum_i\gamma_i\right)}\frac{\prod_i\Gamma\left(n_i+\gamma_i\right)/\Gamma(\gamma_i)}{\prod_i(p_0)_i^{n_i}}
\end{equation}
which allows us to calculate the Bayes factor without any sampling. 
\section{Application to Kalman Filtering under Model-Form Error}

We demonstrate the method and its implications in this area by examining a spring-mass-damper (SMD) system with nonlinearities in the spring that are unknown to the modeler. Consider the stochastic differential equation
\begin{equation}
    \label{eq:duffing_oscillator}
    \begin{split}
    \begin{pmatrix}
        x'\\v'
    \end{pmatrix}=\begin{pmatrix}
        0&1\\-k&-c
        \end{pmatrix}
        \begin{pmatrix}
            x\\v
    \end{pmatrix}+
    \begin{pmatrix}
        0\\ -\epsilon x^3+ \mathcal{W}
    \end{pmatrix}\\
    x(0)=1 \ , \ v(0)=0
    \end{split}
\end{equation}
where $c$ is the damping strength, $k$ the spring constant, $\epsilon$ a parameter that controls nonlinearity, and $\mathcal{W}$ represents a zero mean white noise process, and an experiment wherein a discrete Kalman filter is used to predict the location of the mass at a sequence of equally spaced time steps with step size $h$ based on the previous location. At time step $t$, the optimal prediction and associated predictive covariance matrix are given by

\begin{equation}
    \label{eq:kalman_filter_prediction}
    \vec{x}_{t+1|t} =\mat{F}\vec{x}_{t|t}
\end{equation}
\begin{equation}
    \label{eq:kalman_filter_predcov}
    \mat{P}_{t+1|t}=\mat{FP}_{t|t}\mat{F}^T+\begin{pmatrix}
        0&0\\0&h^2\sigma_w^2
    \end{pmatrix}
\end{equation}
where $\mat{F}$ is a matrix corresponding to the solution method for the equation in Eq. \eqref{eq:duffing_oscillator} when $\epsilon=0$, for example an Euler or Runge-Kutta method, and $\vec{x}_{t|t}$ and $\mat{P}_{t|t}$ are the posterior mean and covariances after incorporating observations of the state at the previous time step. For more details on Kalman filtering, see \cite{welch1995introduction}.

In the Kalman filter, the predictive distribution of the next observation $\vec{z}_{t+1}$ given the previous time steps is given by $\mathcal{N}(\vec{x}_{t+1|t},\mat{P}_{t+1|t})$. The calibration of the model is checked by estimating how often the true position of the spring falls within $I_{\alpha}$ for $\alpha=(0.05,0.25,0.50,0.75,0.95,0.99)$ over a range of $\epsilon\in[0,1]$. Due to the setup of the problem, where everything but the nonlinearity is known exactly by the filter, we expect to see perfect calibration at $\epsilon=0$ and slowly worsening performance as $\abs{\epsilon}$ increases. An example of two simulated trajectories for $\epsilon=0,1$ is shown in Figure \ref{fig:kf_fig}. Visually, it appears that the model is well-calibrated when $\epsilon=0$ and likely poorly calibrated when $\epsilon=1$. The goal now is to determine whether this method gives results that agree with our intuition. 

\begin{figure}
\centering
\begin{subfigure}{0.85\textwidth}
\includegraphics[width=\textwidth]{ 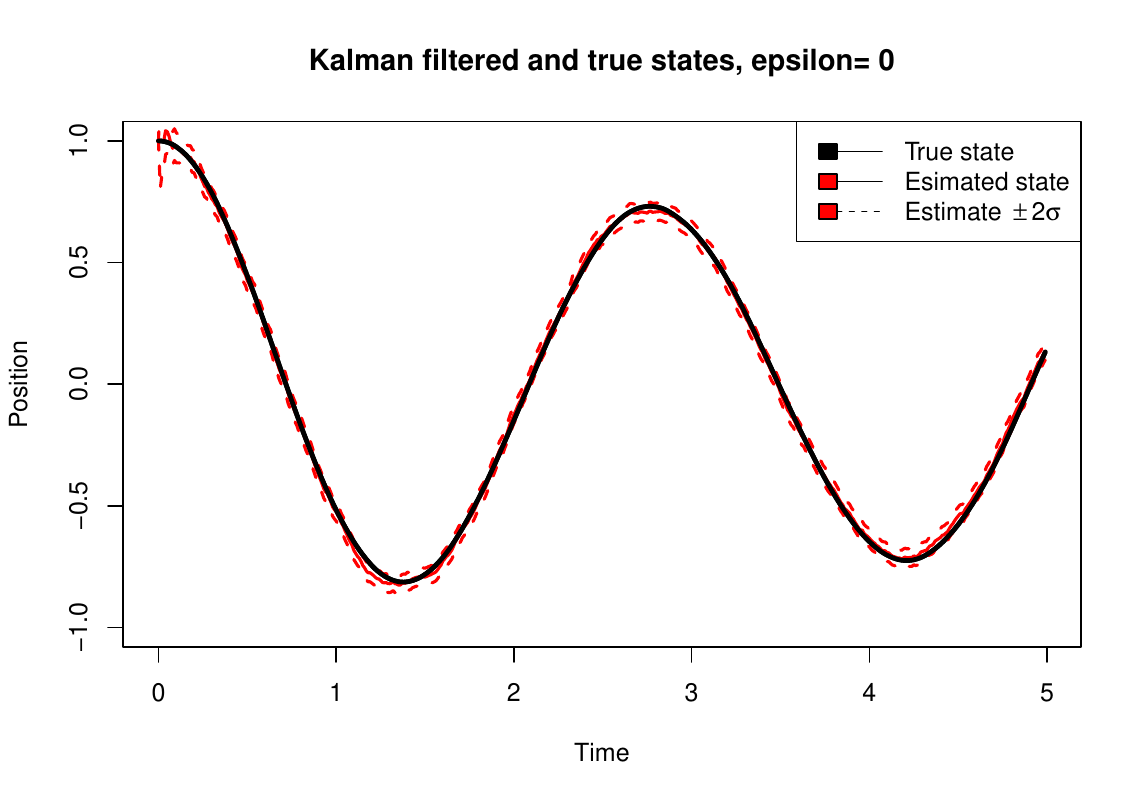}
\caption{$\epsilon=0$}
\label{fig:kf_eps0}
\end{subfigure}
\hfill
\begin{subfigure}{0.85\textwidth}
\includegraphics[width=\textwidth]{ 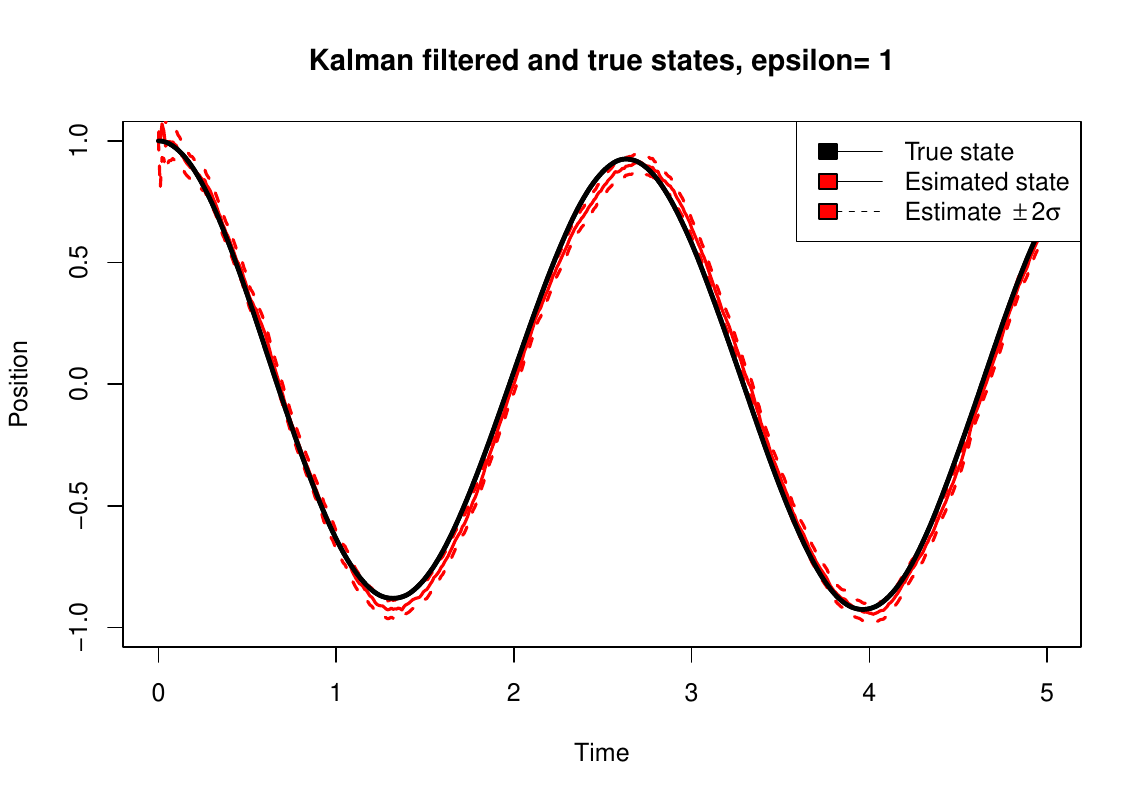}
\caption{$\epsilon=1$}
\label{fig:kf_eps1}
\end{subfigure}
\caption{Simulated trajectories of the spring-mass damper system for $\epsilon=0,1$.}
\label{fig:kf_fig}
\end{figure}

The reliability curves were calculated based on varying numbers of roughly independent observations (separated by 100 time steps). Those plots, for $n=20,200,2000$ samples, are shown in Figure \ref{fig:reliability_plots} for the values of $\epsilon = $ 0 (red), 0.25 (yellow), 0.5 (green), 0.75 (blue) and 1 (purple). At 20 realizations the model has very high uncertainty, but with 200 and more observations it appears as though the model is capable of clearly separating the calibrated and uncalibrated curves. The prior distribution is a Dir($\vec{1}$) distribution. 

\begin{figure}[h!]
    \centering
    \includegraphics[width=\linewidth]{ 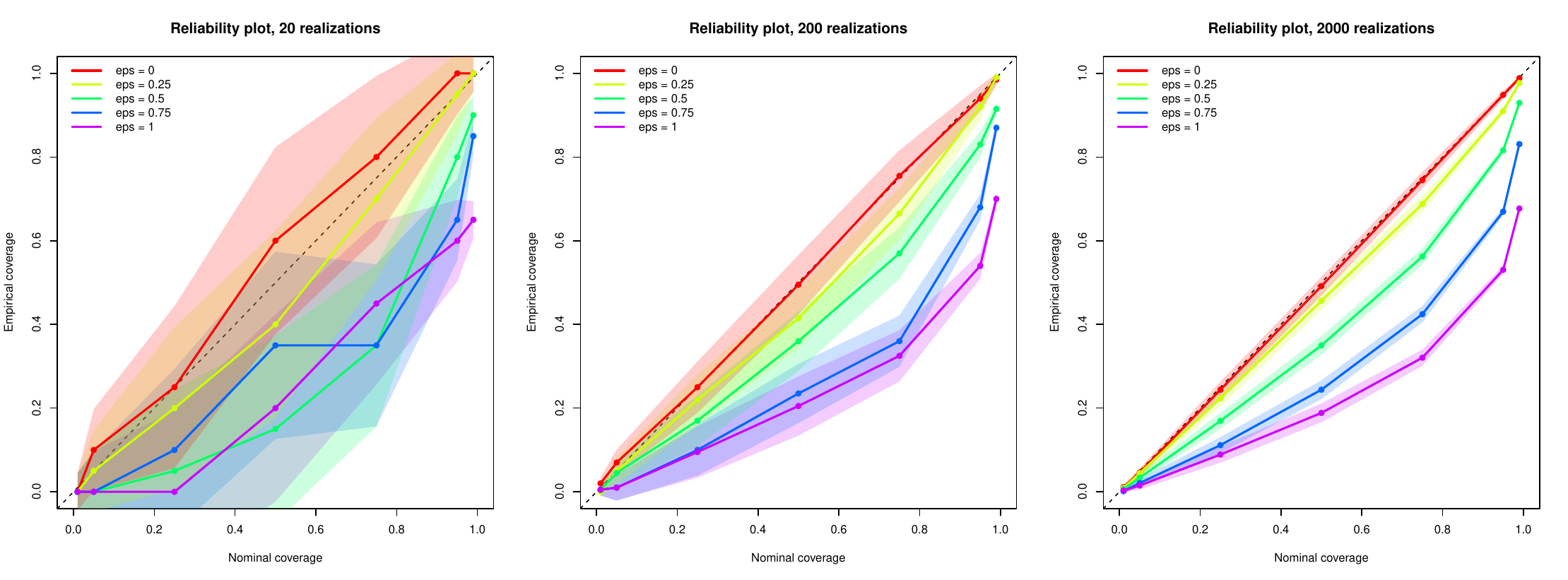}
    \caption{Reliability plots for varying $\epsilon$ based on 20, 200, and 2000 observations along with uncertainty derived from the posterior distribution of the cumulative coverage probabilities. The values of $\epsilon$ are 0 (red), 0.25 (yellow), 0.5 (green), 0.75 (blue) and 1 (purple).}
    \label{fig:reliability_plots}
\end{figure}

To examine the performance of the method, we have checked the probability that the method gives a Bayes Factor greater than 10, termed the detection probability, for each combination of $\epsilon$ and $n$ the number of calibration samples. A threshold of $BF_{10}>10$ was selected to be consistent with "strong" evidence of miscalibration as described in \cite{kass1995bayesfactors}. Each simulation setting was replicated 1000 times, and the resulting probabilities are reported in Figure \ref{fig:detection_prob}. All sample sizes do not report miscalibration when $\epsilon=0$, with only one of the 1000 trials reported as miscalibrated when $n=20$ and none for $n=200,2000$. When $n=20,200$ the method struggles to detect miscalibration when $\epsilon=0.25$, and when $\epsilon=0.5$ the detection probabilities are 0.09,0.64, and 0.98 for $n=20,200,2000$. However, for large $\epsilon$, and when $n$ increases, the method is able to detect miscalibration at all scales. The performance at $\epsilon=0$ suggests that this is a far more stringent criterion for detecting miscalibration than, in a frequentist test, rejecting for $p$ values less than 0.05.

\begin{figure}[h!]
    \centering
    \includegraphics[width=\linewidth]{ 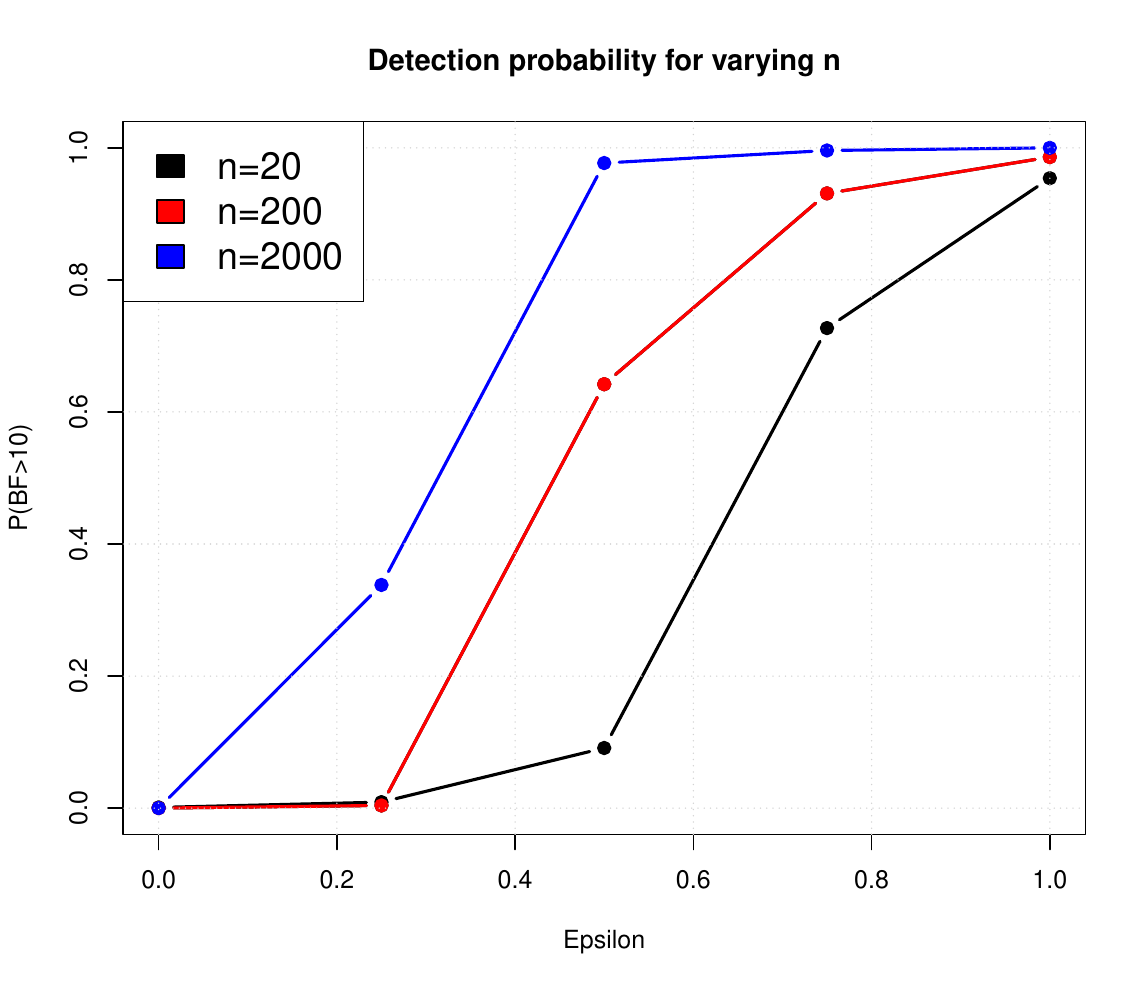}
    \caption{Probability of a Bayes factor larger than 10 for varying $\epsilon$ and $n$. We see that, for larger $\epsilon$, even with only 20 simulations the model is able to detect miscalibration easily.}
    \label{fig:detection_prob}
\end{figure}

One advantage of this method is the relative ease with which it can be extended into one that tests calibration relative to a specific region, for example allowing a tolerance for miscalibration or allowing predictive intervals that are too wide but not too narrow. Here we present an example. Suppose that, at coverage level $\alpha$, we are content to know that the actual coverages are within $\alpha\pm \frac{1}{10}\sqrt{\alpha(1-\alpha)}$ (See Figure \ref{fig:tol_region} for a plot of the region and empirical reliability curves with 200 sample points). This means that mild miscalibration is not of concern to us, and we are more concerned with catastrophic miscalibration. 

\begin{figure}[h!]
    \centering
    \includegraphics[width=\linewidth]{ 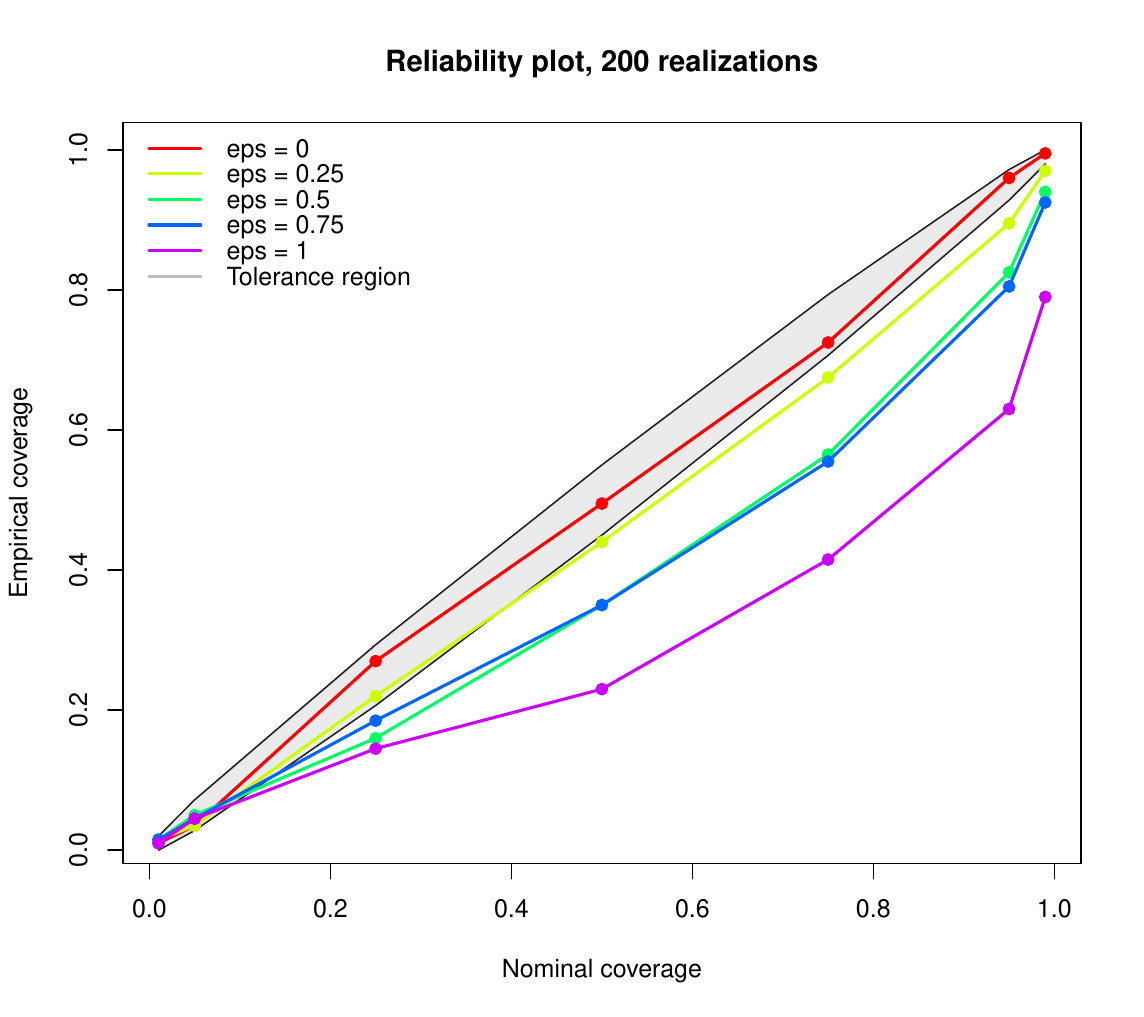}
    \caption{Reliability plots for the Kalman filtering problem using a tolerance region. Uncertainty bands are removed for ease of visualization. }
    \label{fig:tol_region}
\end{figure}

The detection probabilities are shown in Figure \ref{fig:detect_tol}. They tell a similar story to the previous model, with the exception that in this case the filter is considered acceptably calibrated by all models when $\epsilon=0.25$. This points to the advantage of this method, since we are often not concerned with mild miscalibration it may be worthwhile to use a method that is designed to ignore it. Another difference, not shown here, is that the Bayes factors are relatively much smaller in this case than when doing a two sided test. This is because of the size of the tolerance region, meaning that many more observed datasets may be consistent with it. 

\begin{figure}[h!]
    \centering
    \includegraphics[width=\linewidth]{ 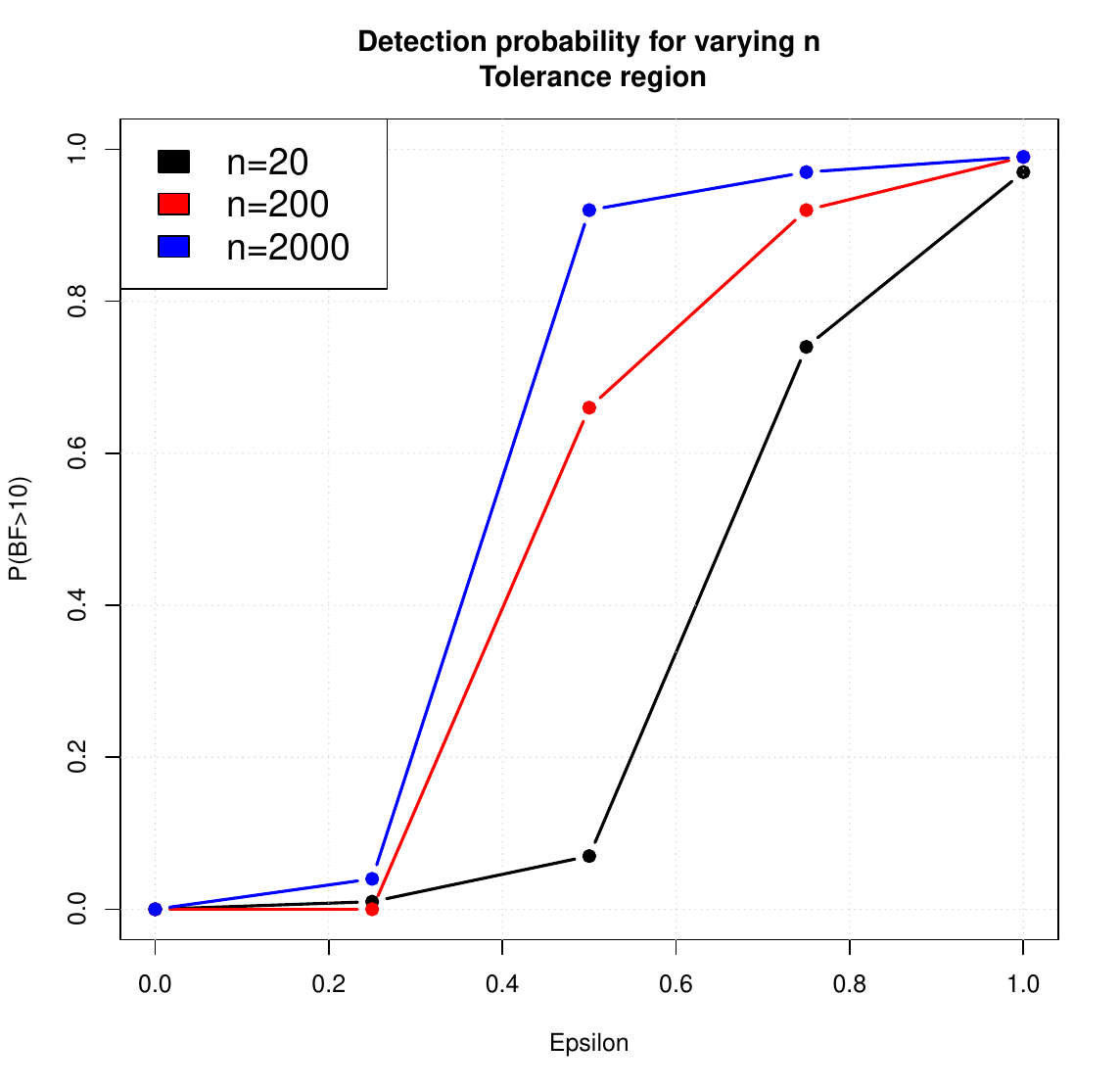}
    \caption{Detection probabilities under the same design as before, but detecting departure of the reliability curve from a tolerance band rather than from perfect calibration. }
    \label{fig:detect_tol}
\end{figure}

\section{Discussion}
This work presents a straightforward Bayesian approach to calibration assessment applicable when the full predictive distribution is unavailable and only confidence intervals or predictive sets are presented. By considering estimation of the reliability curve as a problem of estimating the probability that the true parameter appears in a set of disjoint, exhaustive sets, we derive a Dirichlet-Multinomial model for the reliability curve and Bayes factors for determining calibration both relative to a perfectly calibrated model and relative to a hypothesis that the reliability curve lies within some region of the x-y plane. This approach is simple and more informative than reporting empirical coverage fractions, and allows estimation of the number of samples required to achieve a desired level of uncertainty.

While this work is being put on hold immediately, several directions for future work remain. The most immediate is the application of the method to quantile regression, which is planned for the future. 
Next is the development of frequentist equivalents of the proposed tests. Tests for undercoverage can be described in terms of the stochastic dominance of one reliability curve over another, which connects naturally to order-restricted inference and chi-bar-square asymptotics \citep{silvapulle2005constrained}, and a likelihood ratio test can be 
derived using the pool-adjacent-violators algorithm for the constrained MLE. We leave the full development of these tests, including their 
asymptotic properties and finite-sample bootstrap implementations, to future work. Approximate two-sided tests, such as the one based on the $T_{trunc}$ statistic described in \cite{balakrishnan2018hypothesis_testing_review}, can be used for similar purposes. 

A second direction concerns the nonparametric limit of the Dirichlet-Multinomial model mentioned in Remark 1. As the number of nominal levels $d \to \infty$ with the levels becoming dense in 
$(0,1)$, the Dirichlet model converges to a Dirichlet process prior on the function $g$ in the representation $f(\alpha) = \int_0^\alpha g(t)\,dt$ of the reliability curve. This connects the finite-level calibration assessment to the problem of 
estimating the full reliability curve as a random function, which would yield a posterior credible band for $f(\alpha)$ at every $\alpha \in (0,1)$ simultaneously. This may be doable via probabilistic 
maximum entropy approaches \citep{bretthorst}, Gaussian process priors on $\log g$, or isotonic regression, for example. 

Finally, while the present paper focuses on scalar parameters and univariate credible intervals, the extension to multivariate parameters and joint credible regions is practically important for 
applications in data assimilation and uncertainty quantification for complex systems. Under the assumption that the distribution for the parameter vector is multivariate Gaussian, 
\begin{equation}
    \vec{\theta}|X\sim\mathcal{N}(\vec{\theta}_*,\mathbf{\Sigma})
\end{equation}
calibration implies that the squared Mahalanobis distance 
\begin{equation}
    D^2=\left(\vec{\theta}_0-\vec{\theta}_*\right)^T\mathbf{\Sigma}^{-1}\left(\vec{\theta}_0-\vec{\theta}_* \right)
\end{equation}
where $\vec{\theta}_0$ is the true parameter vector, has a $\chi^2$ distribution with$p=\text{len}(\vec{\theta})$ degrees of freedom. Consequently, calibration can be assessed by checking whether $\alpha = \Pr_{\vec{\theta}_0,X}\left(D^2\le \chi^2_{p,\alpha}\right)$ using the same method. This is also easily applied to methods that report these uncertainty ellipses automatically rather than full distributions. We leave these extensions to future work.

\vspace{0.5in}
\bibliography{refs_reliabilityplots}

@article{kass1995bayesfactors,
  title={Bayes factors},
  author={Kass, Robert E and Raftery, Adrian E},
  journal={Journal of the american statistical association},
  volume={90},
  number={430},
  pages={773--795},
  year={1995},
  publisher={Taylor \& Francis},
  url={https://www.jstor.org/stable/2291091}
}

@book{silvapulle2005constrained,
  title={Constrained statistical inference: Inequality, order and shape restrictions},
  author={Silvapulle, Mervyn J and Sen, Pranab Kumar},
  year={2005},
  publisher={John Wiley \& Sons}
}

@book{jeffreys1998theory,
  title={The theory of probability},
  author={Jeffreys, Harold},
  year={1998},
  publisher={OuP Oxford}
}

@article{jarosz2014odds,
  title={What are the odds? A practical guide to computing and reporting Bayes factors},
  author={Jarosz, Andrew F and Wiley, Jennifer},
  journal={The Journal of Problem Solving},
  volume={7},
  number={1},
  pages={2},
  year={2014},
  publisher={Purdue University Press}
}

@inproceedings{bretthorst,
  title={The maximum entropy method of moments and Bayesian probability theory},
  author={Bretthorst, G Larry},
  booktitle={AIP Conference proceedings},
  volume={1553},
  number={1},
  pages={3--15},
  year={2013},
  organization={American Institute of Physics}
}

@article{gneiting2007probabilistic,
  title={Probabilistic forecasts, calibration and sharpness},
  author={Gneiting, Tilmann and Balabdaoui, Fadoua and Raftery, Adrian E},
  journal={Journal of the Royal Statistical Society Series B: Statistical Methodology},
  volume={69},
  number={2},
  pages={243--268},
  year={2007},
  publisher={Oxford University Press}
}

@article{brocker2007increasing,
  title={Increasing the reliability of reliability diagrams},
  author={Br{\"o}cker, Jochen and Smith, Leonard A},
  journal={Weather and forecasting},
  volume={22},
  number={3},
  pages={651--661},
  year={2007}
}

@misc{valentin2026recipescalibrationcheckssafetycritical,
      title={Recipes for Calibration Checks in Safety-Critical Applications}, 
      author={Romeo Valentin},
      year={2026},
      eprint={2604.26479},
      archivePrefix={arXiv},
      primaryClass={stat.ME},
      url={https://arxiv.org/abs/2604.26479}, 
}

@misc{welch1995introduction,
  title={An introduction to the Kalman filter},
  author={Welch, Greg and Bishop, Gary},
  year={1995}
}

@article{gebetsberger2017estimation,
   author = "Manuel Gebetsberger and Jakob W. Messner and Georg J. Mayr and Achim Zeileis",
      title = "Estimation Methods for Nonhomogeneous Regression Models: Minimum Continuous Ranked Probability Score versus Maximum Likelihood",
      journal = "Monthly Weather Review",
      year = "2018",
      publisher = "American Meteorological Society",
      address = "Boston MA, USA",
      volume = "146",
      number = "12",
      doi = "10.1175/MWR-D-17-0364.1",
      pages=      "4323 - 4338",
      url = "https://journals.ametsoc.org/view/journals/mwre/146/12/mwr-d-17-0364.1.xml"
}

@article{balakrishnan2018hypothesis_testing_review,
  title={Hypothesis testing for high-dimensional multinomials: A selective review},
  author={Balakrishnan, Sivaraman and Wasserman, Larry},
  year={2018},
  journal="The Annals of Statistics"
}

@InProceedings{pmlr-v80-kuleshov18a,
  title = 	 {Accurate Uncertainties for Deep Learning Using Calibrated Regression},
  author =       {Kuleshov, Volodymyr and Fenner, Nathan and Ermon, Stefano},
  booktitle = 	 {Proceedings of the 35th International Conference on Machine Learning},
  pages = 	 {2796--2804},
  year = 	 {2018},
  editor = 	 {Dy, Jennifer and Krause, Andreas},
  volume = 	 {80},
  series = 	 {Proceedings of Machine Learning Research},
  month = 	 {10--15 Jul},
  publisher =    {PMLR},
  url = 	 {https://proceedings.mlr.press/v80/kuleshov18a.html}}

@article{jiang2011smooth,
  title={Smooth isotonic regression: a new method to calibrate predictive models},
  author={Jiang, Xiaoqian and Osl, Melanie and Kim, Jihoon and Ohno-Machado, Lucila},
  journal={AMIA Summits on Translational Science Proceedings},
  volume={2011},
  pages={16},
  year={2011}
}

@article{resin2026hierarchies,
  title={Hierarchies of Calibration: Classification meets Regression},
  author={Resin, Johannes and Yang, Lu and Gneiting, Tilmann},
  journal={arXiv preprint arXiv:2606.03245},
  year={2026}
}

@article{deperettiGraphicalMethodsInvestigating2010,
  title = {Graphical Methods for Investigating the Finite-Sample Properties of Confidence Regions},
  author = {De Peretti, Christian and Siani, Carole},
  year = 2010,
  journal = {Computational Statistics \& Data Analysis},
  volume = {54},
  number = {2},
  pages = {262--271},
  issn = {01679473},
  doi = {10.1016/j.csda.2009.09.012},
  urldate = {2026-10-03},
  copyright = {https://www.elsevier.com/tdm/userlicense/1.0/},
  langid = {english}
}

@article{JMLR:v27:24-1139,
  title = {Simulation-Based Calibration of Uncertainty Intervals under Approximate Bayesian Estimation},
  author = {Savitsky, Terrance D. and Gershunskaya, Julie},
  year = 2026,
  journal = {Journal of Machine Learning Research},
  volume = {27},
  number = {20},
  pages = {1--32}
}
\newpage
\appendix
\section{The Joint Distribution of the $c_i$}\label{app:cidist}
Earlier, we make the argument that, due to the fact that the categories $A_i$ are disjoint and exhaustive, the data are described by a multinomial distribution and therefore it is more natural to work with the "gap probabilities" $q_i$. However, in a typical setting the question of interest is, for example, "How often does my 95\% credible interval actually contain the truth?" rather than "How often does the truth fall within the 95\% credible set but not the 75\% credible set?". For this reason, we can turn back to the definition of the gap probabilities in Eq. \eqref{eq:c_def} to derive the joint distribution of $c_i$. Notably, the constraint that $c_i\le c_{i+1}$ means that they are not independent. 
\begin{prop}
The joint posterior density of $(c_1, \ldots, c_m)$ on the ordered region
$\{0 < c_1 < c_2 < \cdots < c_m < 1\}$ is
\begin{equation}
    p(c_1, \ldots, c_d \mid n_1,n_2,...,n_{d+1}) =
    \frac{\Gamma\!\left(\sum_{i=1}^{d+1} a_i\right)}{\prod_{i=1}^{d+1} \Gamma(a_i)}
    \prod_{i=1}^{d+1} (c_i - c_{i-1})^{a_i - 1}
\end{equation}
where $c_0=0$ and $c_{d+1}=1$ for $i = 1, \ldots, d+1$. 
\end{prop}
\begin{proof}
    Recall that $a_i=n_i+\gamma_i$. This follows from Eq. \eqref{eq:c_def}. The Jacobian of the transformation from the $\text{Dir}(a_1,...,a_{d+1})$ distribution is given by
    \begin{equation}
        \label{eq:c_xform_jacobian}
        J = \begin{pmatrix}
            1&0&0&...&0\\
            1&1&0&...&0\\
            \vdots&\vdots&\vdots&\vdots&\vdots\\
            1&1&1&\vdots&1
        \end{pmatrix}
    \end{equation}
    which has unit determinant. Making the appropriate substitution gives the result.
\end{proof}

The marginal posteriors are the same Beta distributions that would be
obtained by treating $c_i$ independently with binomial likelihoods as done in \cite{valentin2026recipescalibrationcheckssafetycritical}
using counts $k_{\alpha_i} = \sum_{j=1}^i n_j$, namely a $\text{Beta}(\sum_{j=1}^ia_i,\sum_{j=i+1}^{d+1}a_i )$ distribution. The multinomial--Dirichlet model therefore recovers the standard pointwise Beta posteriors as marginals while additionally
specifying the dependence structure between the $c_i$. For the marginal distribution of $c_i$ we see that the variance of $c_i$ given the data is
\begin{equation}
    \sigma^2 = VAR(c_i|x)=\frac{\left(k_{\alpha_i}+\sum_{j=1}^i\gamma_i\right)\left(n-k_{\alpha_i}+\sum_{j=i+1}^{d+1}\gamma_i\right)}{(n+\sum_i\gamma_i)^2(n+1+\sum_i\gamma_i)}
\end{equation}
Assuming $k_{\alpha_i}>>\sum\gamma_i$ and that the model is well-calibrated we can approximate this as
\begin{equation}
    \label{eq:stdev_approx}
    \sigma \approx \sqrt{\frac{\alpha_i(1-\alpha_i)}{n}} \le \frac{1}{2\sqrt{n}}
\end{equation}
which gives a heuristic for the error in the coverage estimation as a function of $n$ and allows us to select $n$ to achieve an acceptably accurate estimate of the coverage at level $\alpha_i$.

\end{document}